\documentclass{llncs}
\usepackage[utf8]{inputenc}
\usepackage{amsfonts,amssymb,amsmath,amscd,latexsym}   
\usepackage{gastex}
\usepackage{graphicx}
\usepackage{xcolor}
\usepackage{subfig}
\usepackage{caption}
\usepackage{multirow}
\usepackage{comment}
\usepackage{array}
\newcolumntype{C}[1]{>{\centering\arraybackslash}p{#1}}

\def\abs#1{\ensuremath{\lvert #1\rvert}}

\makeatletter
\DeclareRobustCommand\sfrac[1]{\@ifnextchar/{\@sfrac{#1}}%
                                            {\@sfrac{#1}/}}
\def\@sfrac#1/#2{\leavevmode\scalebox{.9}{\kern.1em\raise.5ex
         \hbox{$\m@th\mbox{\fontsize\sf@size\z@
                           \selectfont#1}$}\kern-.1em
         /\kern-.15em\lower.25ex
          \hbox{$\m@th\mbox{\fontsize\sf@size\z@
                            \selectfont#2}$}}}
\DeclareRobustCommand\numfrac[1]{\@ifnextchar/{\@numfrac{#1}}%
                                            {\@numfrac{#1}}}
\def\@numfrac#1{\leavevmode \hbox{$\m@th\mbox{\fontsize\sf@size\z@
                           \selectfont#1}$}}
\makeatother

\newcommand{\nat}{\mathbb N}
\newcommand{\real}{\mathbb R}

\newcommand{\tuple}[1]{\langle #1 \rangle}
\newcommand{\dist}{{\cal D}}

\renewcommand{\u}{{\sf u}}    

\newcommand{\M}{{\cal M}}      
\newcommand{\N}{{\cal N}}

\newcommand{\Epsilon}{{\cal E}}     
      
\newcommand{\Supp}{{\sf Supp}}

\newcommand{\straa}{\sigma} \newcommand{\Straa}{\Sigma}
 
\newcommand{\val}{\mathrm{val}}

\newcommand{\win}[2]{\langle \! \langle 1 \rangle \! \rangle_{\mathit{#2}}^{\mathit{#1}}}
\newcommand{\winsure}[1]{\langle \! \langle 1 \rangle \! \rangle_{\mathit{sure}}^{\mathit{#1}}}
\newcommand{\winas}[1]{\langle \! \langle 1 \rangle \! \rangle_{\mathit{almost}}^{\mathit{#1}}}
\newcommand{\winlim}[1]{\langle \! \langle 1 \rangle \! \rangle_{\mathit{limit}}^{\mathit{#1}}}
\newcommand{\winbound}[1]{\langle \! \langle 1 \rangle \! \rangle_{\mathit{bounded}}^{\mathit{#1}}}
\newcommand{\winpos}[1]{\langle \! \langle 1 \rangle \! \rangle_{\mathit{positive}}^{\mathit{#1}}}

\newcommand{\Act}{{\sf A}}
\newcommand{\APre}{{\sf APre}}
\newcommand{\Pre}{{\sf Pre}}

\let\epsilon\varepsilon
\let\emptyset\varnothing

\newenvironment{longversion}{}{}
\newenvironment{shortversion}{}{}
\begin{document}
\pagestyle{plain}

\excludecomment{shortversion}


\title{{\bf Bounds for Synchronizing\\ Markov Decision Processes}}

\author{Laurent~Doyen\inst{1} \and Marie~van~den~Bogaard\inst{2}}

\institute{CNRS \& LMF, ENS Paris-Saclay, France \and LIGM, Université Gustave Eiffel}

\maketitle

\begin{abstract}
We consider Markov decision processes with synchronizing objectives,
which require that a probability mass of $1-\epsilon$ accumulates in a designated 
set of target states, either once, always, infinitely often, or always
from some point on,  where $\epsilon = 0$ for sure synchronizing, 
and $\epsilon \to 0$ for almost-sure and limit-sure synchronizing.

We introduce two new qualitative modes of synchronizing, 
where the probability mass should be either \emph{positive}, or \emph{bounded}
away from~$0$. They can be viewed as dual synchronizing objectives.
We present algorithms and tight complexity results for the problem of deciding
if a Markov decision process is positive, or bounded synchronizing, and we provide explicit
bounds on $\epsilon$ in all synchronizing modes.
In particular, we show that deciding positive and bounded synchronizing
always from some point on, is coNP-complete. 
\end{abstract}


\section{Introduction}\label{sec:intro}
Markov decision processes (MDP) are finite-state probabilistic systems with 
controllable (non-deterministic) choices. 
They play a central role
in several application domains for practical purpose~\cite{HBMDP,Puterman}, 
and in theoretical computer science as a basic model 
for the analysis of stochastic transition systems~\cite{BK08,CY95}.

In the traditional state-based semantics, 
we consider the \emph{sequences of states} that form a path in the underlying graph of the MDP.
When a control policy (or strategy) for the non-deterministic choices is fixed, 
we obtain a purely stochastic process that induces a probability measure 
over sets of paths~\cite{BK08,CH12}. 

In the more recent distribution-based semantics, 
the outcome of a stochastic process is a sequence of distributions
over states~\cite{BRS02,KVAK10}. This alternative semantics has received
some attention recently for theoretical analysis of probabilistic bisimulation~\cite{HKK14}
and is adequate to describe large populations of agents~\cite{DMS19,CFO20} with applications 
in system biology~\cite{KVAK10,AGV18}. 
The behaviour of an agent is modeled as an MDP with some state space $Q$, 
and a large population of identical agents is described by a (continuous) 
distribution $d: Q \to [0,1]$ that gives the fraction $d(q)$ of agents in the 
population that are in each state $q \in Q$.
The control problem is to construct a strategy for the agents
that guarantees a specified global outcome of the agents,
defined in terms of \emph{sequences of distributions}. Specifications of interest 
include safety objectives~\cite{AGV18} and synchronization objectives~\cite{DMS19}.
A distribution is $p$-synchronized in a set $T$ of states if 
it assigns to the states in $T$ a mass of probability at least $p$.
Synchronization objectives require
that $p$-synchronized distributions occur in the outcome sequence,
either, at some position (eventually), at all positions (always),
infinitely often (weakly), or always from some point on (strongly),
where synchronization is sure winning for $p=1$, and almost-sure 
or limit-sure winning for $p$ arbitrarily close to $1$~\cite{DMS19}.

Consider eventually synchronizing as an illustration.
Formally, denoting by $d_i^{\straa}(T)$ the probability
mass in a set $T$ under strategy $\straa$ at position $i$
in a given MDP, the three winning modes for eventually synchronizing 
correspond to the following three possible orders of the quantifiers:
\begin{itemize}
\item
$\forall \epsilon>0 \cdot \exists \straa \cdot \exists i: d_i^{\straa}(T) \geq 1 - \epsilon$,
for limit-sure winning,
\item
$\exists \straa \cdot \forall \epsilon>0 \cdot \exists i: d_i^{\straa}(T) \geq 1 - \epsilon$,
for almost-sure winning,
\item
$\exists \straa \cdot \exists i \cdot \forall \epsilon>0: d_i^{\straa}(T) \geq 1 - \epsilon$,
for sure winning.
\end{itemize}

Note that the formula $\forall \epsilon>0: d_i^{\straa}(T) \geq 1 - \epsilon$
is equivalent to $d_i^{\straa}(T) = 1$ in the case of sure winning.
%
Defining the value of a strategy $\straa$ as 
$\val(\straa) = \sup_i d_i^{\straa}(T)$,
the question for limit-sure winning is analogous to the 
cutpoint isolation problem for value $1$, i.e. whether
the value $1$ can be approached arbitrarily closely~\cite{Rabin63,BS20}.
Previous work~\cite{DMS19} shows that the above three questions are PSPACE-complete, 
and presents a construction of the (existentially quantified) strategy~$\straa$
when one exists.

\begin{figure}[t]
\begin{center}
    \begin{picture}(38,18)(0,0)

\node[Nmarks=ir,iangle=180](n0)(10,4){$q_{0}$}
\node[Nmarks=n](n1)(30,4){$q_1$}

\drawedge(n0,n1){$\frac{1}{2}$}
\drawloop[ELside=l,loopCW=y, loopangle=90, loopdiam=5](n0){$\frac{1}{2}$}
\drawloop[ELside=l,loopCW=y, loopangle=90, loopdiam=5](n1){}

\end{picture}
\end{center}
 \caption{A Markov chain that is positively but not boundedly synchronizing (for 
all modes except eventually). \label{fig:mc}}
\end{figure}
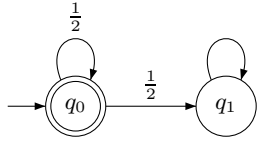

In this paper, we consider dual synchronization objectives obtained
either by taking the negation of the synchronization objectives, or 
by replacing the existential quantifier on strategies by a universal quantifier.
\smallskip 

1.~Taking the negation corresponds to the control player having no 
strategy to satisfy the synchronization objective. In that case, we 
show that a more precise information can be derived, namely bounds 
on the value of $\epsilon$, which is existentially quantified, 
and we construct explicit values for the four synchronizing modes. 
These values give bounds on the isolation distance of the value $1$.
For instance, the negation of limit-sure eventually synchronizing in $T$ is
given by the formula:
$$\exists \epsilon > 0 \cdot \forall \straa \cdot \forall i: d^{\straa}_{i}(T) \leq 1-\epsilon.$$
We show that the statement holds for a value $\epsilon = \epsilon_e(n,\alpha,\alpha_0)$ 
that depends on the number $n$ of states of the MDP,
the smallest positive probability $\alpha$ in the transitions of the MDP,
and the smallest positive probability $\alpha_0$ in the initial distribution $d_0$ 
(see Theorem~\ref{theo:epsilon-eventually}).
The most interesting case is when limit-sure weakly synchronizing does not hold,
that is: 
$$\exists \epsilon > 0  \cdot \forall \straa \cdot 
\exists N \cdot \forall i \geq N: d^{\straa}_{i}(T) \leq 1-\epsilon.$$

Given the value $\epsilon = \epsilon_w$ that satisfies this condition (see Theorem~\ref{theo:epsilon-weakly}),
the value of $N$ can be arbitrarily large (depending on the 
strategy $\straa$). Nevertheless, we can effectively construct a constant $N_w$
such that, for all strategies $\straa$, in the sequence $(d^{\straa}_{i})_{i \in \nat}$
there are at most $N_w$ distributions that are $(1-\epsilon_w)$-synchronized in~$T$.
\smallskip

2.~Replacing the existential strategy quantifier by a universal quantifier 
corresponds to an adversarial MDP where all strategies need to satisfy the requirement,
or after taking the negation, to the existence of a strategy that violates 
a dual of the synchronizing requirement.
Note that there is no more alternation of quantifiers on $\epsilon$ and on $\straa$ 
($\forall \epsilon \cdot \forall \straa$ is the same as 
$\forall \straa \cdot \forall \epsilon$), which gives rise to only two
new winning modes in existential form:
\begin{samepage}
\begin{itemize}
\item
$\exists \straa \cdot \exists \epsilon>0 \cdot \forall i: d_i^{\straa}(T) \geq \epsilon$,
that we call \emph{bounded} winning,
\item
$\exists \straa \cdot \forall i \cdot \exists \epsilon > 0: d_i^{\straa}(T) \geq \epsilon$,
that we call \emph{positive} winning (since this is equivalent to 
$\exists \straa \cdot \forall i: d_i^{\straa}(T) > 0$).
\end{itemize}
\end{samepage}

\begin{table}[!t]
\begin{center}
\caption{
Positive and Bounded winning modes for always, strongly, weakly, and eventually synchronizing objectives. 
\label{tab:def-modes}}{
\smallskip
\begin{tabular}{l@{\quad} *{2}{c@{\;}c@{\;}c@{\qquad}}}
 \large{\strut}         & \multicolumn{3}{c}{Always} & \multicolumn{3}{c}{Strongly} \\
\hline
 Positively \large{\strut}    & $\exists \straa$ & $\forall i$                      & $d^{\straa}_i(T) > 0$ 
        	              & $\exists \straa$ & $\exists N \: \forall i \geq N$  & $d^{\straa}_i(T) > 0$ 
 \\
 Boundedly \  \large{\strut} 	& $\exists \straa$ & $\inf_{i}$                 & $d^{\straa}_i(T) > 0$
				& $\exists \straa$ & $\liminf_{i\to\infty}$     & $d^{\straa}_i(T) > 0$ 
 \\
\hline
 \large{\strut}   &\multicolumn{3}{c}{Weakly} & \multicolumn{3}{c}{Eventually} \\
\hline
 Positively \large{\strut}	  
			& $\exists \straa$ & $\forall N \: \exists i \geq N$   & $d^{\straa}_i(T) > 0$ 
			& $\exists \straa$ & $\exists i$                       & $d^{\straa}_i(T) > 0$ 
 \\
 Boundedly \  \large{\strut}	  
				& $\exists \straa$ & $\limsup_{i\to\infty}$  & $d^{\straa}_i(T) > 0$ 
				& $\exists \straa$ & $\sup_{i}$              & $d^{\straa}_i(T) > 0$ 
 \\
\hline
\end{tabular}  
}
\end{center}
\end{table}


Table~\ref{tab:def-modes} presents the analogous definitions of bounded and positive
winning for the four synchronizing modes. It is easy to see that 
for eventually synchronizing, the positive and bounded mode coincide,
while for the other synchronizing modes the positive and bounded modes are distinct, 
already in Markov chains (see \figurename~\ref{fig:mc}). 

We establish the complexity of deciding bounded and positive winning
in the four synchronizing modes, given an MDP and initial distribution (which 
we call the membership problem), and we also construct explicit values for $\epsilon$. 
Adversarial MDPs are a special case of two-player stochastic games~\cite{CH12}
in which only the second player (the adversary of the first player) 
is non-trivial. The results of this paper will be useful for the
analysis of adversarial MDPs obtained from a game
by fixing a strategy of the first player.
The complexity results are summarized in Table~\ref{tab:complexity}.
For positive winning, memoryless winning strategies exist (playing all 
actions uniformly at random is sufficient), and the problem can be solved 
by graph-theoretic techniques on Markov chains.
For bounded winning, the most challenging case is strongly synchronizing,
where we show that a simple form of strategy with memory is winning,
and that the decision problem is coNP-complete. We give
a structural characterization of bounded strongly synchronizing MDPs, 
and show that it can be decided in coNP.
Note that the coNP upper bound is not obtained by guessing a strategy, 
since the coNP lower bound holds in the case of Markov chains where 
strategies play no role. 
\begin{shortversion}
Omitted proofs and additional material can be found in an extended version of this paper~\cite{DSv22}.
\end{shortversion}

\paragraph{Related works.}
The distribution-based semantics of MDPs~\cite{BRS02,KVAK10} has received an increased amount
of attention recently, with works on safety objectives~\cite{AGV18}
and synchronizing objectives~\cite{DMS19} (see also references therein). 
Logic and automata-based frameworks express distribution-based properties,
by allowing different order of the logical quantifiers,
such as $\forall \straa \exists i$ in standard reachability
which becomes $\exists i \forall \straa$ in synchronized reachability~\cite{BRS02,AFA1,CD16}. 
The bounded and positive winning modes introduced in this paper have not been considered before.
They bear some similarity with the qualitative winning modes in concurrent games~\cite{ConcOmRegGames}. 

Applications are found in modeling of large populations of identical agents,
such as molecules, yeast, bacteria, etc.~\cite{KVAK10,AGV18} where the probability 
distributions represent concentrations of each agent in the system. Analogous
models have been considered in a discrete setting where the number of agents
is a parameter~$n$, giving rise to control problems for parameterized systems,
asking if there exists a strategy that brings all $n$ agents synchronously 
to a target state~\cite{BDGG17,CFO20}.


\begin{table}[t]
\begin{center}
\caption{Computational complexity of the membership problem (for eventually 
synchronizing, the positive and bounded modes coincide).\label{tab:complexity}}{
\smallskip
\begin{tabular}{l@{\quad}c@{\quad}c@{\quad}c@{\quad}c}
\large{\strut}            &  Always & Eventually            & Weakly       & Strongly \\
\hline                        
Positively \large{\strut} & coNP-C  & \multirow{2}{*}{NL-C} & NL-C    & coNP-C  \\
Boundedly \large{\strut}  & coNP-C  &                       & NL-C & coNP-C  \\
\hline
\end{tabular}  
}
\end{center}
\end{table}

\section{Definitions}\label{sec:def}

A \emph{probability distribution} over a finite set~$Q$ is a
function $d : Q \to [0, 1]$ such that $\sum_{q \in Q} d(q)= 1$. 
The \emph{support} of~$d$ is the set $\Supp(d) = \{q \in Q \mid d(q) > 0\}$. 
We denote by $\dist(Q)$ the set of all probability distributions over~$Q$. 
Given a set $T\subseteq Q$, let $d(T)=\sum_{s\in T}d(s)$.


A \emph{Markov decision process} (MDP) is a tuple $\M = \tuple{Q,\Act,\delta}$ 
where $Q$ is a finite set of states, $\Act$ is a finite set of labels called actions, 
and $\delta: Q \times \Act \to \dist(Q)$ is a probabilistic transition function.   
A \emph{Markov chain} is an MDP with singleton action set $\abs{\Act} = 1$.
Given a state $q\in Q$ and an action $a \in \Act$, the successor state of $q$
under action $a$ is $q'$ with probability $\delta(q,a)(q')$.

Given $X,Y \subseteq Q$, let 
$$\APre(Y,X) = \{q \in Q \mid \exists a \in \Act: \Supp(\delta(q,a)) \subseteq Y \land \Supp(\delta(q,a)) \cap X \neq \emptyset\},$$
be the set of states from which there is an action to ensure that all successor 
states are in~$Y$ and that with positive probability the successor state is in $X$, 
and for $X = Q$, let 
$$\Pre(Y) = \APre(Y,Q) = \{q \in Q \mid \exists a \in \Act: \Supp(\delta(q,a)) \subseteq Y\}$$
be the set of states from which there is an action to ensure (with probability~$1$)
that the successor state is in~$Y$.  
For $k>0$, let $\Pre^k(Y) = \Pre(\Pre^{k-1}(Y))$ with $\Pre^0(Y) = Y$.
Note that the sequence $\Pre^k(Y)$ of iterated predecessors is ultimately periodic,
precisely there exist $k_1 < k_2 \leq 2^{\abs{Q}}$ such that $\Pre^{k_1}(Y) = \Pre^{k_2}(Y)$.

\paragraph{Strategies.}
A (randomized) \emph{strategy} in $\M$ is a function $\straa: (Q\Act)^{*}Q \to \dist(\Act)$
that, given a finite sequence $\rho = q_0 a_0 q_1 a_1 \dots q_k$,
chooses the next action $a_k$ with probability $\straa(\rho)(a_k)$.
We write $\straa_1 \subseteq \straa_2$ if $\Supp(\straa_1(\rho)) \subseteq \Supp(\straa_2(\rho))$
for all $\rho \in (Q\Act)^{*}Q$.
A strategy $\straa$ is \emph{pure} if for all $\rho \in (Q\Act)^{*}Q$, 
there exists an action $a \in \Act$ such that $\straa(\rho)(a) = 1$. 
In all problems considered in this paper, it is known that pure strategies 
are sufficient~\cite{DMS19}. However, the bounds we provide in case
there is no winning strategy hold for all strategies, pure or randomized.

Given an initial distribution $d_0 \in \dist(Q)$ and a strategy $\straa$ in $\M$,
the probability of a finite sequence $\rho = q_0 a_0 q_1 a_1 \dots q_k$ is 
defined by 
$\Pr^{\straa}_{d_0}(\rho) = d_0(q_0) \cdot \prod_{j=0}^{k-1} 
\straa(q_0 a_0\dots q_j)(a_j) \cdot \delta(q_j,a_j)(q_{j+1}).$
For an initial distribution $d_0$ such that $d_0(q_0) = 1$,
we sometimes write $\Pr^{\straa}_{q_0}(\cdot)$ and say that 
$q_0$ is the initial state. 
We say that $\rho$ is \emph{compatible} with $\straa$ and $d_0$ if 
$\Pr^{\straa}_{d_0}(\rho) > 0$. By extension, an infinite sequence
$\pi \in (Q\Act)^{\omega}$ is compatible with $\straa$ and $d_0$ if 
all prefixes of $\pi$ that end in a state are compatible.
It is standard to extend (in a unique way) $\Pr^{\straa}_{d_0}$ over Borel sets of infinite 
paths in $(Q\Act)^{\omega}$ (called events), by assigning probability $\Pr^{\straa}_{d_0}(\rho)$
to the basic cylinder set containing all infinite paths with prefix $\rho$~\cite{Vardi-focs85,BK08}. 
Given a set $T \subseteq Q$ of target states, and $k \in \nat$, 
we define the following events (sometimes called objectives): 
\begin{itemize}
\item $\Box T = \{q_0 a_0 q_1 \dots \in (Q\Act)^{\omega} \mid \forall i: q_i \in T\}$ 
the safety event of staying in~$T$; 
\item $\Diamond T = \{q_0 a_0 q_1 \dots \in (Q\Act)^{\omega} \mid \exists i: q_i \in T\}$
the event of reaching~$T$; 
\item $\Diamond^{k}\, T = \{q_0 a_0 q_1 \dots \in (Q\Act)^{\omega} \mid q_k \in T \}$ the event
of reaching~$T$ after exactly $k$~steps;
\end{itemize}

A distribution~$d_0$ is \emph{almost-sure winning} for an event $\Omega$ 
if there exists a strategy $\straa$ such that $\Pr^{\straa}_{d_0}(\Omega) = 1$,
and \emph{limit-sure winning} if $\sup_{\straa} \Pr^{\straa}_{d_0}(\Omega) = 1$, 
that is the event $\Omega$ can be realized with probability arbitrarily close to~$1$.
Finally $d_0$ is \emph{sure winning} for $\Omega$ if there exists 
a strategy $\straa$ such that all paths compatible with $\straa$ and $d_0$ 
belong to $\Omega$.

Safety and reachability events are dual, in the sense that $\Diamond T$
and $\Box(Q \setminus T)$ form a partition of $(Q\Act)^{\omega}$. 
It is known for safety objectives $\Box T$ that the three 
winning regions (sure, almost-sure winning, and limit-sure winning) 
coincide in MDPs, and for reachability objectives $\Diamond T$, 
almost-sure and limit-sure winning coincide~\cite{deAlfaro97}.
It follows that if the negation of almost-sure reachability holds, that is
$\Pr^{\straa}_{d_0}(\Diamond T) < 1$ for all strategies $\straa$, 
then equivalently \mbox{$\inf_{\straa} \Pr^{\straa}_{d_0}(\Box(Q \setminus T)) > 0$} (note the strict inequality),
the probability mass that remains always outside $T$ can be bounded.
An explicit bound can be obtained from the classical characterization 
of the winning region for almost-sure reachability~\cite{CY95}.

\begin{lemma}\label{lem:almost-sure-reach-MDP}
If a distribution $d_0$ is not almost-sure winning for a reachability
objective $\Diamond T$ in an MDP $\M$, then for all strategies $\straa$, 
for all $i \geq 0$, 
we have 
$Pr^{\straa}_{d_0}(\Diamond^{i}\, T) \leq 1 - \alpha_0 \cdot \alpha^{n}$
where $n = \abs{Q}$ is the number of states and
$\alpha$ the smallest positive probability in $\M$,
and $\alpha_0 = \min \{d_0(q) \mid q \in \Supp(d_0)\}$ is the smallest 
positive probability in the initial distribution $d_0$.
\end{lemma}

\begin{longversion}

\begin{proof}\label{proof:almost-sure-reach-MDP}
The set of almost-sure winning states for the reachability objective 
$\Diamond T$ can be computed 
by graph-theoretic algorithms~\cite{CY95}, and can be expressed succinctly 
by the fixpoint $\mu$-calculus formula 
$\varphi_{{\sf AS}} = \nu Y.\mu X.~\!T \cup \APre(Y,X)$.
We briefly recall the interpretation of  $\mu$-calculus formulas~\cite{BS07,BW18}
based on Knaster-Tarski theorem.
Given a monotonic function $\psi: 2^Q \to 2^Q$
(i.e., such that $X \subseteq Y$ implies $\psi(X) \subseteq \psi(Y)$),
the expression $\nu Y. \psi(Y)$ is the (unique) greatest fixpoint of $\psi$,
which can be computed as the limit of the sequence $(Y_i)_{i \in \nat}$ defined by
$Y_0 = Q$, and $Y_{i} = \psi(Y_{i-1})$ for all $i \geq 1$.
Dually, the expression $\mu X.~\psi(X)$ is the (unique) least fixpoint of $\psi$,
and the limit of the sequence $(X_i)_{i \in \nat}$ defined by
$X_0 = \emptyset$, and $X_{i} = \psi(X_{i-1})$ for all $i \geq 1$.
If $\abs{Q} = n$, then it is not difficult to see that the 
limit of those sequences is reached after at most $n$ iterations, 
$X_{n} = X_{n+1}$ and $Y_{n} = Y_{n+1}$.

Intuitively, the formula $\varphi_{{\sf AS}}$ computes the largest set $S$ of states
such that every state $q \in S$ has a strategy to ensure reaching $T$ with 
positive probability, while at the same time staying in $S$ with probability~$1$.
It follows that $S$ is the set of all states from which there exists a strategy 
to reach $T$ with probability~$1$.

Now consider the states $q_0 \not\in \varphi_{{\sf AS}}$ that are not almost-sure 
winning, and let the rank of $q_0$ be the least integer $i$ such that
$q_0 \in Y_{i}$ and $q_0 \not\in Y_{i+1}$.
Let $\Omega_{i} = Y_i \setminus Y_{i+1}$ for all $i \geq 0$ 
be the set of states of rank~$i$, and for $\mbox{$\sim$}\!\in\!\{<,\leq\}$
let $\Omega_{\sim i} = \bigcup_{j\sim i} \Omega_{j}$. 
It is easy to show by induction that $Q \setminus Y_i = \Omega_{<i}$ 
since $\Omega_{<i} = (Y_0 \setminus Y_{1}) \cup (Y_1 \setminus Y_{2}) \cup \dots 
 \cup (Y_{i-1} \setminus Y_{i})$ and $Y_0 = Q$. Moreover since $T \subseteq Y_i$
we have $\Omega_{i} \cap T = \emptyset$.

For $q_0 \in \Omega_{i}$, we show that  
$\Pr^{\straa}_{q_0}(\Box \Omega_i \cup \Diamond \Omega_{<i}) \geq \alpha$ for all
strategies $\straa$ (where $\alpha$ is the smallest positive probability in the transitions of the MDP).
Since $q_0 \not\in Y_{i+1} = \mu X.~T \cup \APre(Y_i,X)$, we have 
$q_0 \not\in \APre(Y_{i},Y_{i+1})$.
Then for all actions $a \in \Act$ we have:
$$\text{either } \Supp(\delta(q_0,a)) \not\subseteq Y_{i}, \text{ or }
  \Supp(\delta(q_0,a)) \cap Y_{i+1} = \emptyset,$$
\noindent that is 
$$\text{either } \Supp(\delta(q_0,a)) \cap \Omega_{<i} \neq \emptyset,
  \text{ or } \Supp(\delta(q_0,a)) \subseteq \Omega_{\leq i}.$$ 
It follows in both cases
that if not all successors of $q_0$ on action $a$ are in $\Omega_{i}$,
then at least one of them is in $\Omega_{<i}$, which entails that 
$\Pr^{\straa}_{q_0}(\Box \Omega_i \cup \Diamond \Omega_{<i}) \geq \alpha$.
For $i=0$, since $\Omega_{<0} = \emptyset$, we have 
$\Pr^{\straa}_{q_0}(\Box \Omega_0) = 1$ and thus $\Pr^{\straa}_{q_0}(\Box (Q \setminus T)) = 1$ 
for all $q_0 \in \Omega_{0}$ and all strategies $\straa$.  
Inductively, if $\Pr^{\straa}_{q_0}(\Box (Q \setminus T)) \geq p_i$ for all 
$q_0 \in \Omega_{< i}$ and all strategies $\straa$, then   
$\Pr^{\straa}_{q_0}(\Box (Q \setminus T)) \geq \alpha \cdot p_i$ for all $q_0 \in \Omega_{i}$
and all strategies $\straa$. 

It follows that if $q_0 \not\in \varphi_{{\sf AS}}$ is not almost-sure winning
for the reachability objective $\Diamond T$, then 
$q_0 \in \Omega_{<n}$ and $\Pr^{\straa}_{q_0}(\Box (Q \setminus T)) \geq \alpha^n$
for all strategies $\straa$, which entails the result of the lemma.  
\qed
\end{proof}
\end{longversion}

In Lemma~\ref{lem:almost-sure-reach-MDP} it is crucial to notice that the bound 
$\alpha_0 \cdot \alpha^{n}$ is independent of the number $i$ of steps.

\paragraph{Synchronizing objectives.}
We consider MDPs as generators of sequences of probability distributions 
over states~\cite{KVAK10}.
Given an initial distribution $d_0 \in \dist(Q)$ and a strategy $\straa$ in $\M$,
the sequence $\M^{\straa} = (\M^{\straa}_i)_{i \in \nat}$ of probability distributions (from $d_0$,
which we assume is clear from the context)
is defined by $\M^{\straa}_i(q) = Pr^{\straa}_{d_0}(\Diamond^{i}\, \{q\})$ for all $i \geq 0$ 
and $q \in Q$.
Hence, $\M^{\straa}_i$ is the probability distribution over states after $i$ steps
under strategy $\straa$. Note that $\M^{\straa}_0 = d_0$. 

Informally, synchronizing objectives require that the probability of some set~$T$
of states tends to $1$ in the sequence $(\M^{\straa}_i)_{i\in \nat}$,
either always, once, infinitely often, or always after some point~\cite{DMS19}. 
Given a target set $T \subseteq Q$,
we say that a probability distribution $d$ is \emph{$p$-synchronized} in $T$ 
if $d(T) \geq p$ (and strictly $p$-synchronized in $T$ if $d(T) > p$), and 
that a sequence $d_0 d_1 \dots$ of probability distributions is:
\begin{itemize}
\item[$(a)$] \emph{always $p$-synchronizing} if $d_i$ is $p$-synchronized (in $T$) for all $i \geq 0$;
\item[$(b)$] \emph{event(ually) $p$-synchronizing} if $d_i$ is $p$-synchronized (in $T$) for some $i \geq 0$;
\item[$(c)$] \emph{weakly $p$-synchronizing} if $d_i$ is $p$-synchronized (in $T$) for infinitely many $i$'s;
\item[$(d)$] \emph{strongly $p$-synchronizing} if $d_i$ is $p$-synchronized (in $T$) for all but finitely many $i$'s.
\end{itemize}

Given an initial distribution $d_0$, 
we say that for the objective of \{always, eventually, weakly, strongly\} synchronizing
from~$d_0$, the MDP $\M$ is:
\begin{itemize}
\item \emph{sure winning} if there exists a strategy $\straa$ such that
the sequence $\M^{\straa}$ from $d_0$ 
is \{always, eventually, weakly, strongly\} $1$-synchronizing in $T$;
\item \emph{almost-sure winning} if there exists a strategy $\straa$ such that 
for all $\epsilon>0$ the sequence $\M^{\straa}$ from $d_0$  is 
\{always, eventually, weakly, strongly\} $(1-\epsilon)$-synchronizing in $T$;
\item \emph{limit-sure winning} if for all $\epsilon>0$, there exists a strategy $\straa$
such that the sequence $\M^{\straa}$ from $d_0$  is 
\{always, eventually, weakly, strongly\} $(1-\epsilon)$-synchronizing in $T$;
\end{itemize}

For $\lambda \in \{always, event, weakly, strongly\}$, 
we denote by $\winsure{\lambda}(T)$ the \emph{winning region} 
defined as the set of initial distributions $d_0$ such that
$\M$ is sure winning for $\lambda$-synchronizing in $T$ 
(in this notation, we assume that $\M$ is clear from the context). 
We define analogously the winning regions $\winas{\lambda}(T)$ and $\winlim{\lambda}(T)$
of almost-sure and limit-sure winning distributions.

It is known that for all winning modes, only the support of the
initial distributions is relevant, that is for every winning region 
$W = \win{\lambda}{\mu}(T)$ (where $\mu \in \{sure, almost, limit\}$),
for all distributions $d,d'$, if $\Supp(d) = \Supp(d')$, then 
$d \in W$ if and only if $d' \in W$~\cite{DMS19}.
Therefore, in the sequel we sometimes write $S \in \win{\lambda}{\mu}(T)$,
which can be read as any distribution $d$ with support $S$ is in $\win{\lambda}{\mu}(T)$. 
For each synchronizing mode $\lambda$ and winning mode $\mu$,
the membership problem asks to decide,
given an MDP $\M$, a target set $T$, and a set $S$, whether
$S \in \win{\lambda}{\mu}(T)$.

\begin{figure}[t]
\begin{center}
    \begin{picture}(78,18)(0,0)

\node[Nmarks=i,iangle=180](n0)(10,4){$q_{0}$}
\node[Nmarks=n](n1)(30,4){$q_1$}
\node[Nmarks=r](n2)(50,4){$q_2$}
\node[Nmarks=n](n3)(70,4){$q_3$}

\drawedge(n0,n1){$\frac{1}{2}$}
\drawloop[ELside=l,loopCW=y, loopangle=90, loopdiam=5](n0){$\frac{1}{2}$}

\drawedge(n1,n2){$b$}
\drawloop[ELside=l,loopCW=y, loopangle=90, loopdiam=5](n1){$a$}

\drawedge(n2,n3){$a,b$}
\drawloop[ELside=l,loopCW=y, loopangle=90, loopdiam=5](n3){$a,b$}

\end{picture}
\end{center}
 \caption{An MDP where $\lbrace q_0 \rbrace \in \winlim{event}(\{q_2\}).$\label{fig:almost-limit-eventually-differ}}
\end{figure}

\smallskip
Consider the MDP in \figurename~\ref{fig:almost-limit-eventually-differ}, 
with initial state $q_0$ and target set $T = \{q_2\}$.
The probability mass cannot loop through $q_2$ and therefore, 
it is immediate that the MDP is neither always, nor weakly, nor strongly $(1-\epsilon)$-synchronizing,
thus $\{q_0\} \not\in \winas{\lambda}(T)$ for $\lambda = always, weakly, strongly$,
and thus also $\{q_0\} \not\in \winsure{\lambda}(T)$.

For eventually synchronizing in $q_2$, at every step, half of the probability 
mass in $q_0$ stays in $q_0$ while the other half is sent to $q_1$. 
Thus, the probability mass in $q_0$ tends to $0$ but is strictly positive 
at every step, and the MDP is not sure eventually synchronizing, $\{q_0\} \not\in \winsure{event}(T)$.
In state $q_1$, action $a$ keeps the probability mass in $q_1$, while action $b$ 
sends it to the target state $q_2$.  
If action $b$ is never chosen, then $q_2$ is never reached, and whenever $b$ is 
chosen, a strictly positive probability mass remains in $q_0$, thus the MDP is 
not almost-sure eventually synchronizing, $\{q_0\} \not\in \winas{event}(T)$.
On the other hand, for every $\epsilon > 0$, the strategy that plays $a$ in $q_1$ 
for $k$ steps such that $\frac{1}{2^k} < \epsilon$, and then plays $b$, 
is winning for eventually $(1-\epsilon)$-synchronizing in $T$. 
Thus the MDP is limit-sure eventually synchronizing, $\{q_0\} \in \winlim{event}(T)$.

The MDP in \figurename~\ref{fig:inf-mem} is also limit-sure eventually synchronizing in $\{q_2\}$.
As the probability mass is sent back to $q_0$ from $q_2$,
the MDP is even almost-sure weakly (and eventually) synchronizing,
using a strategy that plays action~$a$ in $q_1$ for $k$ steps to
accumulate probability mass $1-\frac{1}{2^k}$ in $q_1$, then plays action~$b$
and repeats the same pattern for increasing values of $k$.

\begin{figure}[t]
\begin{center}
    \begin{picture}(58,23)

\node[Nmarks=i,iangle=180](n0)(10,9){$q_{0}$}
\node[Nmarks=n](n1)(30,9){$q_1$}
\node[Nmarks=r](n2)(50,9){$q_2$}

\drawedge[ELdist=1](n0,n1){$\frac{1}{2}$}
\drawloop[ELside=l, loopCW=y, loopangle=90, loopdiam=5](n0){$\frac{1}{2}$}

\drawedge(n1,n2){$b$}
\drawloop[ELside=r,loopCW=n, loopangle=90, loopdiam=5](n1){$a$}

\drawedge[ELpos=50, ELdist=.5, ELside=r, curvedepth=10](n2,n0){$a,b$}

\end{picture}
\end{center}
 \caption{An MDP where $\lbrace q_0 \rbrace \in \winas{weakly}(\{q_2\}).$ \label{fig:inf-mem}}
\end{figure}


\paragraph{End-components.}

Given a state $q \in Q$ and a set $S \subseteq Q$,
let $\Act_S(q)$ be the set of all actions $a \in \Act$ such 
that $\Supp(\delta(q,a)) \subseteq S$. 
A \emph{closed} set in an MDP is a set $S \subseteq Q$ such 
that $\Act_S(q) \neq \emptyset$ for all $q \in S$. 
A set $S \subseteq Q$ is an {\em end-component}~\cite{deAlfaro97,BK08} if 
(i)~$S$ is closed, and 
(ii)~the graph $(S,E_S)$ is strongly connected
where $E_S=\{(q,q') \in S \times S \mid \delta(q,a)(q') > 0 \text{ for some } a \in \Act_S(q)\}$
denote the set of edges given the actions. 
In the sequel, end-components should be considered \emph{maximal}, that is 
such that no strict superset is an end-component.
We denote by $\Epsilon$ the union of all end-components, 
and for $q \in \Epsilon$, we denote by $\Epsilon(q)$ the maximal end-component
containing $q$.
A fundamental property of end-components is that under arbitrary strategies,
with probability~$1$ the set of states visited infinitely often along a 
path is an end-component.

\begin{lemma}[\cite{CY95,deAlfaro97}]\label{lem:ec}
Let $\M$ be an MDP.   
For all strategies $\straa$, we have $\liminf_{i \to \infty} \M^{\straa}_{i}(\Epsilon) = 1$.
\end{lemma}


\paragraph{Tracking counter in MDPs.}
It will be useful to track the number of steps (modulo a given number $r$)
in MDPs. 
Given a number $r \in \nat$, define the MDP $\M \times [r] = \tuple{Q_r, \Act, \delta_r}$
where $Q_r = Q \times \{r-1,\dots,1,0\}$ and $\delta_r$ is defined as follows,
for all $\tuple{q,i}, \tuple{q',j} \in Q_r$ and $a \in \Act$:
$$\delta_r(\tuple{q,i},a)(\tuple{q',j}) = 
\begin{cases}
\delta(q,a)(q') & \text{ if } j=i-1 \!\!\mod r,\\
0 & \text{ otherwise.}
\end{cases}
$$

For a distribution $d \in \dist(Q)$ and $0 \leq t < r$, 
we denote by $d \times \{t\}$ the distribution defined, for all $q \in Q$, 
by $d \times \{t\}(\tuple{q,i}) = d(q)$ if $t=i$, and 
$d \times \{t\}(\tuple{q,i}) = 0$ otherwise.
Given a finite sequence $\rho = q_0 a_0 q_1 a_1 \dots q_n$ in $\M$, 
and $0 \leq t < r$, there is a corresponding sequence 
$\rho' = \tuple{q_0,k_0} a_0 \tuple{q_1,k_1} a_1 \dots \tuple{q_n, k_n}$ in 
$\M \times [r]$ where $k_0 = t$ and $k_{i+1} = k_i - 1 \mod r$ for all $0 \leq i < n$.
Since the sequence $\rho'$ is uniquely defined from $\rho$ and $t$, 
there is a clear bijection between the paths in $\M$ starting in $q_0$
and the paths in $\M \times [r]$ starting in $\tuple{q_0,t}$.
In the sequel, we freely omit to apply and mention this bijection.
In particular, we often consider that a strategy $\straa$ in $\M$
can be played directly in $\M \times [r]$.

\smallskip 

\begin{figure}[t]
\begin{center}
    \begin{picture}(60,31)

\node[Nmarks=i,iangle=180](n0)(10,15){$q_{0}$}
\node[Nmarks=n](n1)(30,25){$q_1$}
\node[Nmarks=r](n2)(50,25){$q_2$}
\node[Nmarks=r](n3)(30,5){$q_3$}
\node[Nmarks=n](n4)(50,5){$q_4$}

\drawedge[ELside=l, ELdist=1](n0,n1){$\frac{1}{2}$}
\drawedge[ELside=r, ELdist=1](n0,n3){$\frac{1}{2}$}


\drawedge[ELpos=50, ELdist=.5, ELside=r, curvedepth=-6](n1,n2){}
\drawedge[ELpos=50, ELdist=.5, ELside=r, curvedepth=-6](n2,n1){}

\drawedge[ELpos=50, ELdist=.5, ELside=r, curvedepth=-6](n3,n4){}
\drawedge[ELpos=50, ELdist=.5, ELside=r, curvedepth=-6](n4,n3){}


\end{picture}
\end{center}
 \caption{A Markov chain with two periodic end-components. \label{fig:periodicMDP}}
\end{figure}
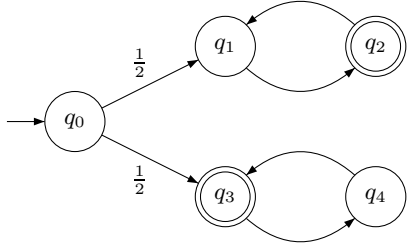

Consider the MDP in \figurename~\ref{fig:periodicMDP} (which is in fact a Markov chain), 
with initial state $q_0$ and target set $T = \{q_2,q_3\}$. 
There are two end-components,
$S_1 = \lbrace q_1, q_2 \rbrace$ and $S_2 = \lbrace q_3, q_4 \rbrace$. 
Although both $S_1$ and $S_2$ are sure eventually synchronizing in $T$ (from 
$q_1$ and $q_3$ respectively), the uniform distribution over $\{q_1,q_3\}$ 
is not even limit-sure eventually synchronizing in $T$. 

\section{Eventually synchronizing}\label{sec:event}

In the rest of this paper, fix an MDP $\M = \tuple{Q,\Act,\delta}$ 
and let $n = \abs{Q}$ be the size of $\M$, and let $\alpha$ be the 
smallest positive probability in the transitions of $\M$.
We first consider eventually synchronizing, that is synchronization
must happen once. 



\subsection{Sure winning}\label{sec:event-sure}
We recall the following characterization of the sure-winning region for 
eventually synchronizing~\cite[Lemma~7]{DMS19}:
$d \in \winsure{event}(T)$ 
if and only if there exists $k \geq 0$ such that $\Supp(d) \subseteq \Pre^{k}(T)$. 
Intuitively, from all states in $\Pre^{i}(T)$, there exists an action to ensure
that the successor is in $\Pre^{i-1}(T)$ (for all $i>0$), 
and therefore there exists a strategy to ensure that the probability mass in 
$\Pre^{k}(T)$ reaches $T$ in exactly $k$ steps.
Now, if $q_0 \not\in \Pre^{k}(T)$, then for all sequences of actions $a_0,\dots,a_{k-1}$ 
there is a path $q_0 a_0 q_1 a_1 \dots q_k$ of length $k$ that ends in $q_k \in Q \setminus T$.
It is easy to derive the following result from this characterization.

\begin{lemma}\label{lem:sure-eventually}
If $d_0 \not\in \winsure{event}(T)$ is not sure eventually synchronizing in~$T$, 
then for all strategies $\straa$, for all $i \geq 0$, we have:
$$\M^{\straa}_{i}(T) \leq 1- \alpha_0 \cdot \alpha^i$$ 
where $\alpha_0$ is the 
smallest positive probability in~$d_0$.
\end{lemma}

Note that the bound $1 - \alpha_0 \cdot \alpha^i$ tends to $1$ as $i \to \infty$,
which is unavoidable since  MDPs that are not sure eventually synchronizing 
may be almost-sure eventually synchronizing~\cite{DMS19}.  
The following variant of Lemma~\ref{lem:sure-eventually} will be useful
in the sequel.

\begin{remark}\label{rem:sure-eventually-support}
If $\{q_0\} \not\in \winsure{event}(T)$ and we take 
$\alpha'_0 = d_0(q_0)$, then from the initial distribution $d_0$ we have, 
for all strategies $\straa$, for all $i \geq 0$:
$$\M^{\straa}_{i}(T) \leq 1 - \alpha'_0 \cdot \alpha^i.$$
\end{remark}

\subsection{Limit-sure winning}

If the MDP $\M$ is not limit-sure winning for eventually synchronizing
in $T$, then the probability in $Q \setminus T$ is bounded away 
from $0$ in all distributions in $\M^{\straa}$ (for all strategies $\straa$).
We give an explicit bound $\epsilon_e$ as follows.

\begin{theorem}\label{theo:epsilon-eventually}
Given an initial distribution $d_0$, 
let $\alpha_0$ be the smallest positive probability in~$d_0$, 
and let $\epsilon_e = \alpha_0 \cdot \alpha^{(n+1)\cdot 2^{n}}$.
If $d_0 \not\in \winlim{event}(T)$ is not limit-sure
eventually synchronizing in $T$, 
then for all strategies $\straa$, for all $i \geq 0$, we have:
$$\M^{\straa}_{i}(T) \leq 1-\epsilon_e,$$
that is, no distribution in $\M^{\straa}$ is strictly 
$(1-\epsilon_e)$-synchronizing in $T$.
\end{theorem}

\footnotetext[1]{The results of~\cite[Lemma~11~\&~12]{DMS19} consider
a more general definition of limit-sure synchronizing, 
where the support of the $(1-\epsilon)$-synchronizing 
distribution is required to have its support contained 
in a given set $Z$. We release this constraint by taking $Z = Q$.}

\begin{proof} 
We recall the characterization\footnotemark[1]
of~\cite[Lemma~11]{DMS19} for limit-sure synchronizing in an arbitrary set $T$. 
For all $k\geq 0$, we have 
$$\winlim{event}(T) = \winsure{event}(T) \cup \winlim{event}(R),$$
where $R = \Pre^{k}(T)$.  

We use this characterization with a specific $k$ (and $R$) as follows.
Consider the sequence of predecessors~$\Pre^{i}(T)$ (for $i=1,2,\dots$),
which is ultimately periodic. Let $0 \leq k < 2^{n}$ and 
$1 \leq r < 2^{n}$ be such that $\Pre^{k}(T) = \Pre^{k+r}(T)$, 
and let $R = \Pre^{k}(T)$. Thus $R = \Pre^{k+r}(T) = \Pre^{r}(R)$.

Since $d_0 \not\in \winlim{event}(T)$, we have:
$$(a)~d_0 \not\in \winsure{event}(T), \ \text{ and }\ (b)~d_0 \not\in \winlim{event}(R).$$
By~$(a)$, it follows from Lemma~\ref{lem:sure-eventually} 
that $\M^{\straa}_{i}(T) \leq 1 - \alpha_0 \cdot \alpha^i$
for all strategies $\straa$ and all $i \geq 0$, which establishes
the bound in the lemma for the first $2^{n}$ steps,
since $\alpha_0 \cdot \alpha^i \geq \epsilon_e$ for all $i \leq 2^{n}$. 

We now recall the characterization\footnotemark[1] of~\cite[Lemma~12]{DMS19} 
for limit-sure synchronizing in the set $R$, which has the property that 
$R = \Pre^{r}(R)$:
$d_0 \in \winlim{event}(R)$ if and only if  
there exists $0 \leq t < r$ such that $d_0 \times \{t\}$ 
is almost-sure winning for the reachability objective 
$\Diamond (R \times \{0\})$ in the MDP $\M \times [r]$.
By~$(b)$, it follows that for all $0 \leq t < r$, the distribution
$d_0 \times \{t\}$ is \emph{not} almost-sure winning for the reachability 
objective $\Diamond (R \times \{0\})$ in the MDP $\M \times [r]$.

Let $\N = \M \times [r]$.
By Lemma~\ref{lem:almost-sure-reach-MDP}, 
from all distributions $d_0 \times \{t\}$ (for all $0 \leq t < r$), 
for all strategies $\straa$ and all $i \geq 0$,
we have:
$$\N^{\straa}_{i}(R \times \{0\}) \leq 
1 - \alpha_0 \cdot \alpha^{\abs{Q_r}} = 1- \alpha_0 \cdot  \alpha^{n \cdot 2^{n}}.$$
Since this holds for all $t=0,\dots,r-1$, 
we conclude that $\M^{\straa}_{i}(Q \setminus R) \geq \alpha_0 \cdot \alpha^{n \cdot 2^{n}}$
in the original MDP $\M$ from $d_0$, for all strategies $\straa$ and all $i \geq 0$.

Since $R = \Pre^{k}(T)$, it follows from 
Lemma~\ref{lem:sure-eventually} and Remark~\ref{rem:sure-eventually-support}
that, if at step $i$ a mass of probability $p$ is outside $R$,
then at step $i+k$ a mass of probability at least $p \cdot \alpha^k$ 
is outside $T$. Hence we have 
$\M^{\straa}_{i+k}(Q \setminus T) \geq \alpha_0 \cdot \alpha^{n \cdot 2^{n}} \cdot \alpha^k 
\geq \alpha_0 \cdot \alpha^{(n+1)\cdot 2^{n}}$
for all strategies $\straa$ and for all $i \geq 0$, 
which implies $\M^{\straa}_{i}(Q \setminus T) \geq \alpha_0 \cdot \alpha^{(n+1)\cdot 2^{n}}$
for all $i \geq 2^{n}$ (since $k < 2^{n}$).

Combining the results for $i \leq 2^{n}$ and for $i \geq 2^{n}$,
we get $\M^{\straa}_{i}(T) \leq 1-\epsilon_e$ for all $i \geq 0$, which
concludes the proof. \qed
\end{proof} 

Theorem~\ref{theo:epsilon-eventually} also gives a sufficient condition
that can be used as an alternative to~\cite[Lemma~11]{DMS19}
to show that an MDP is limit-sure eventually synchronizing.
This will be useful in the proof of our main result (Theorem~\ref{theo:epsilon-weakly}).

A variant of Theorem~\ref{theo:epsilon-eventually} is obtained
by observing that if $d_0 \not\in \winlim{event}(T)$, there 
exists a set $S_0 \subseteq \Supp(d_0)$ such that 
$S_0 \not\in \winlim{event}(T)$. It may be that $S_0$
is a strict subset of $\Supp(d_0)$, and then it is sufficient 
to consider $\alpha_0$ as the smallest positive probability of $d_0$
on $S_0$. 

\begin{remark}\label{rem:epsilon-eventually-support}
If $S_0 \not\in \winlim{event}(T)$ and $S_0 \subseteq \Supp(d_0)$,
then we can define $\alpha_0$ by $\min \{d_0(q) \mid q \in S_0\}$
in the bound $\epsilon_e$ of Theorem~\ref{theo:epsilon-eventually}.
\end{remark}

\subsection{Almost-sure winning}
A simple argument shows that the almost-sure winning region 
for eventually synchronizing consists of the union of the
sure winning region for eventually synchronizing and
the almost-sure winning region for weakly synchronizing~\cite[Section 5.1.2]{Shi14},
that is $\winas{event}(T) = \winsure{event}(T) \cup \winas{weakly}(T)$.

It follows that if $d_0 \not\in \winas{event}(T)$, then both 
$d_0 \not\in \winsure{event}(T)$ and $d_0 \not\in \winas{weakly}(T)$,
and we can use both the results of 
Lemma~\ref{lem:sure-eventually} and Theorem~\ref{theo:epsilon-weakly}.

\section{Weakly synchronizing}\label{sec:weakly}

We now consider weakly synchronizing, which intuitively requires that 
synchronization happens infinitely often.

\subsection{Sure winning}

We recall the following characterization of the sure-winning region for 
weakly synchronizing~\cite[Lemma~18]{DMS19}:
for all distributions $d_0 \in \dist(Q)$, we have $d_0 \in \winsure{weakly}(T)$ 
if and only if there exists a set $S \subseteq T$ such that 
$\Supp(d) \subseteq \Pre^k(S)$ for some $k \geq 0$,
and $S \subseteq \Pre^r(S)$ for some $r \geq 1$.

\begin{lemma}\label{lem:sure-weakly}
If $d_0 \not\in \winsure{weakly}(T)$ is not sure weakly synchronizing in~$T$, 
then for all strategies $\straa$,
in the sequence $\M^{\straa}$ there are at most $2^{n}$ distributions that 
are $1$-synchronized in $T$, that is
$\M^{\straa}_{i}(T) = 1$ for at most $2^{n}$ values of $i$.
\end{lemma}

\begin{longversion}
\begin{proof}
Consider the support $S_j = \Supp(\M^{\straa}_{i_j})$ (where $i_1 < i_2 < \dots$)
of the distributions in the sequence $\M^{\straa}$ that are 
$1$-synchronized in $T$. We show by contradiction, 
that all $S_j$ are distinct, which establishes the result
since they can take at most $2^{n}$ different values.

If $S_j = S_k$ for some $j < k$,
then it is easy to show that $\M$ is sure-winning for eventually
synchronizing in $S = S_j$ from $d_0$, and $\M$ is also 
sure-winning for eventually synchronizing in $S = S_k$ from $S = S_j$,
which implies by the characterization of the sure-winning region
for eventually synchronizing (Section~\ref{sec:event-sure}) that
$\Supp(d) \subseteq \Pre^{i_j}(S)$ and $S \subseteq \Pre^{i_k-i_j}(S)$,
thus $d_0 \in \winsure{weakly}(T)$, which contradicts the assumption
of the lemma. \qed
\end{proof}
\end{longversion}

\subsection{Limit-sure and almost-sure winning}

The winning region for limit-sure and almost-sure weakly synchronizing
coincide~\cite[Theorem~7]{DMS19}. Therefore, in the sequel we treat 
them interchangeably.
We recall the following characterization of almost-sure weakly
synchronizing. 
 
\begin{lemma}[\mbox{\cite[Lemma~23,Theorem~7]{DMS19}}]\label{lem:ls-weakly}
For all distributions $d_0$, the following equivalence holds:
$d_0 \in \winas{weakly}(T)$
if and only if there exists a set $T' \subseteq T$ such that:  
$$d_0 \in \winlim{event}(T') \ \text{ and }\  T' \in \winlim{event}(\Pre(T')).$$
\end{lemma}

The condition in Lemma~\ref{lem:ls-weakly} ensures that from $d_0$
almost all the probability mass (namely $1-\epsilon$ for arbitrarily small $\epsilon > 0$)
can be injected in a set $T'$ of target 
states in $0$ or more steps, and that from $T'$ almost all the probability 
mass can be injected in $\Pre(T')$, thus also in $T'$ (but, in at least $1$ step).
Intuitively, by successively halving the value of $\epsilon$ one can 
construct a strategy that ensures almost all the 
probability mass loops through $T'$, 
thus a limit-sure weakly synchronizing strategy (which is equivalent
to the existence of an almost-sure weakly synchronizing strategy). 

If $d_0 \not\in \winas{weakly}(T)$ is not almost-sure weakly synchronizing,
we use Lemma~\ref{lem:ls-weakly} to show that for all sets $T' \subseteq T$,
if $d_0 \in \winlim{event}(T')$ is limit-sure eventually synchronizing in $T'$,
then $T'$ is not limit-sure eventually synchronizing in $\Pre(T')$
(i.e., $T' \not\in \winlim{event}(\Pre(T'))$). This implies that 
a bounded number of distributions in the sequence $\M^{\straa}$
can be $(1-\epsilon)$-synchronized in~$T$ (for sufficiently small $\epsilon$).
We now state the main result of this section.

\begin{theorem}\label{theo:epsilon-weakly}
Given an initial distribution $d_0$,
let $\alpha_0$ be the smallest positive probability in~$d_0$,
and let $\epsilon_w = \alpha_0 \cdot \frac{\alpha^{(n+2)\cdot 4^{n}}}{{n}^{2^{n}+1}}$
and $N_w = 2^{n}$.

If $d_0 \not\in \winas{weakly}(T)$ is not almost-sure weakly synchronizing in~$T$, 
then for all strategies $\straa$, in the sequence $\M^{\straa}$ at most $N_w$ 
distributions are strictly $(1-\epsilon_w)$-synchronized in $T$,
that is $\M^{\straa}_{i}(T) > 1-\epsilon_w$ for at most $N_w$ values of $i$.
\end{theorem}

\begin{proof} 
%
Given the assumption of the lemma, 
we show the following statement by induction on $k = 0,1,\dots, 2^n$:
if there are $k$ distributions in $\M^{\straa}$ that 
are strictly $(1-\epsilon_w)$-synchronized in $T$, then there 
exist $k$ \emph{distinct} nonempty sets $T_1, \dots, T_k \subseteq T$
such that no distribution after those $k$ distributions 
in $\M^{\straa}$ is strictly $(1-\epsilon_w)$-synchronized in $T_j$ 
(for all $1 \leq j \leq k$). 

For $k = 2^{n}$, one of the sets $T_j$ is equal to $T$ and 
it follows that at most $2^{n}$ distributions in $\M^{\straa}$ can be 
$(1-\epsilon_w)$-synchronized in $T$, which concludes the base case.
For the proof by induction, we use the bound $\epsilon_e$ of 
Theorem~\ref{theo:epsilon-eventually}.
Let $F = (n+1)\cdot 2^{n}$ (thus $\epsilon_e = \alpha_0 \cdot \alpha^{F}$)
and for $k=0,1,\dots$ define $z_k = \frac{\alpha_0}{n} \cdot \left(\frac{\alpha^{F+1}}{n}\right)^k$.
We prove a slightly stronger statement:
for $k = 0,1,\dots, 2^n$, if there are $k$ positions $i_1< i_2 < \ldots < i_k$
such that the distributions $\M^{\straa}_{i_j}$ ($j=1,\dots,k$) 
are strictly $(1-\epsilon_w)$-synchronized in $T$, then there 
exist $k$ distinct nonempty sets $T_1, \dots, T_k \subseteq T$
such that no distribution after position $i_j$
in $\M^{\straa}$ is strictly $(1- z_j \cdot \alpha^{F+1})$-synchronized in $T_j$ 
(for all $1 \leq j \leq k$).

This statement is indeed stronger since the sequence $z_k$ is decreasing, and 
$\epsilon_w$ was chosen such that $\epsilon_w \leq z_{2^n}$, 
from which it follows that $1-\epsilon_w \geq 1-z_k$
for all $k \leq 2^n$. 

The base case for $k=0$ holds trivially. For the induction case,
assume that the statement holds for a given $k < 2^n$, and show that 
it holds for $k+1$ as follows. If there are $k+1$ positions $i_1< i_2 < \ldots < i_{k+1}$
such that all distributions $d_{j} = \M^{\straa}_{i_j}$ ($j=1,\dots,k+1$) 
are strictly $(1-\epsilon_w)$-synchronized in $T$,
then by the induction hypothesis, no distribution after position $i_j$
in $\M^{\straa}$ is strictly $(1- z_j \cdot \alpha^{F+1})$-synchronized in $T_j$ 
(for all $1 \leq j \leq k$).

Now consider the distribution $d_{k+1}$ at position $i_{k+1}$, 
which is $(1-\epsilon_w)$-synchronized in $T$
and appears after position $i_k$ in $\M^{\straa}$.
We construct the set 
$T_{k+1} = \{q \in T \cap \Supp(d_{k+1}) \mid d_{k+1}(q) > z_{k+1} \}$,
which contains the states in $T$ that carry enough probability mass (namely $z_{k+1}$)
according to $d_{k+1}$.

Note that not all states in $T \cap \Supp(d_{k+1})$ carry 
a probability mass less than $z_{k+1}$:
otherwise, the total mass of $T$ in $d_{k+1}$ would be at most 
$n \cdot z_{k+1} \leq 1-\epsilon_w$ (this inequality holds thanks to
$n\geq 2$), in contradiction with 
$d_{k+1}$ being $(1-\epsilon_w)$-synchronized in $T$.
Therefore $T_{k+1}$ is nonempty.
Hence the set $T_{k+1}$ can be obtained from $T$ by removing at most $n-1$ states
and we have
$$
\begin{cases}
d_{k+1}(T_{k+1}) > 1-\epsilon_w - (n-1)\cdot  z_{k+1} \geq 1 - n \cdot  z_{k+1} = 1 - z_{k}\cdot \alpha^{F+1} & \\
d_{k+1}(q) > z_{k+1} \text{ for all } q \in T_{k+1} & \\
\end{cases}
$$

So $d_{k+1}$ is strictly $(1-z_{k}\cdot \alpha^{F+1})$-synchronized in $T_{k+1}$,
and therefore also strictly $(1-z_{j}\cdot \alpha^{F+1})$-synchronized in $T_{k+1}$
(for all $1 \leq j \leq k$). 
Then, the induction hypothesis implies that the set $T_{k+1}$ is distinct 
from $T_1,\dots,T_{k}$.
Since $1-z_{k}\cdot \alpha^{F+1} \geq 1-z_{0}\cdot \alpha^{F+1} > 1 - \epsilon_e$, 
it follows that $d_{k+1} = \M^{\straa}_{i_{k+1}}$ is 
strictly $(1-\epsilon_e)$-synchronized in $T_{k+1}$,
and by Theorem~\ref{theo:epsilon-eventually}, that 
the initial distribution $d_0$
is limit-sure eventually synchronizing in $T_{k+1}$,
that is $d_0 \in \winlim{event}(T_{k+1})$.

By Lemma~\ref{lem:ls-weakly}, this entails that 
$T_{k+1}$ is not limit-sure eventually synchronizing in $\Pre(T_{k+1})$
(i.e., $T_{k+1} \not\in \winlim{event}(\Pre(T_{k+1}))$),
and by Theorem~\ref{theo:epsilon-eventually}, for all distributions $d$
in $\M^{\straa}$ that occur at or after position $i_{k+1}$, 
we have $d(\Pre(T_{k+1})) \leq 1- z_{k+1} \cdot \alpha^F$
where $z_{k+1} < \min \{d_{k+1}(q) \mid q \in T_{k+1} \cap \Supp(d_{k+1})\}$ 
is a lower bound on the smallest positive probability of a state 
of $T_{k+1}$ in the distribution $d_{k+1}$, taken as the initial distribution 
(see Remark~\ref{rem:epsilon-eventually-support}).
It follows that for all distributions $d$
in $\M^{\straa}$ that occur (strictly) after position $i_{k+1}$, 
we have $d(T_{k+1}) \leq 1- z_{k+1} \cdot \alpha^{F+1}$.
Hence no distribution in $\M^{\straa}$ after $d_{k+1}$
is strictly $(1- z_{k+1} \cdot \alpha^{F+1})$-synchronized, which 
concludes the proof of the induction case. \qed
\end{proof}

\section{Always and Strongly Synchronizing}\label{sec:always-strongly}

\begin{shortversion}
The anaysis of always and strongly synchronizing modes is relatively straightforward,
and we present the bounds in Theorem~\ref{theo:epsilon-always-strongly}.
\end{shortversion}

\begin{longversion}
We summarize the key insights for always and strongly synchronizing (i.e., 
synchronizing always from some point on),
and give the bounds in Theorem~\ref{theo:epsilon-always-strongly}.

\paragraph{Always synchronizing.}
For always synchronizing, all winning modes coincide $\winsure{always}(T) =
\winas{always}(T) = \winlim{always}(T)$ and it is easy to see that sure winning
for always synchronizing is equivalent to sure winning for the safety objective 
$\Box T$~\cite[Lemma~2]{DMS19}.
A standard analysis of MDPs shows that if a sure safety objective $\Box T$ is violated,
then for all \emph{pure} strategies $\straa$ there exists a path compatible with
$\straa$ from some state in the support of the initial distribution
to a state in $Q \setminus T$, of length at most $\abs{Q}$ and thus of probability 
at least $\alpha_0 \cdot \alpha^{n}$ where $\alpha_0$ is the smallest positive
probability in the initial distribution, and $\alpha$ is the smallest positive 
probability in the transitions of $\M$. Under a \emph{randomized} strategy $\straa$,
a probability mass $\alpha_0 \cdot \alpha^{n}$ also reaches $Q \setminus T$
within $n$ steps, but possibly along paths of different lengths.
By a pigeonhole principle, at some length (smaller than $n$)
the probability mass must be at least $\frac{\alpha_0}{n} \cdot \alpha^{n}$ in $Q \setminus T$.

\paragraph{Strongly synchronizing.}
For strongly synchronizing, the almost-sure and limit-sure winning modes 
coincide~\cite[Corollary~28]{DMS19}. On the other hand, the sure and almost-sure 
winning modes are equivalent to, respectively, sure, and almost-sure winning 
for the reachability objective~$\Diamond S$  
where $S$ is the sure-winning region for the safety objective~$\Box T$~\cite[Lemma~27]{DMS19}.
It follows that:

\begin{itemize}
\item if $d_0 \not\in \winsure{strongly}(T)$
is not sure winning for strongly synchronizing in~$T$, 
then for all \emph{pure} strategies $\straa$
there exists an infinite path $q_0 a_0 q_1 a_1 \dots$
compatible with $\straa$
from some state $q_0$ in the support of $d_0$
that remains outside $S$, that is $q_i \in Q \setminus S$
for all $i \geq 0$. The corresponding probability mass is at least
$\alpha_0 \cdot \alpha^{i}$, and hence for all \emph{randomized} strategies 
$\straa$, a probability mass at least $\alpha_0 \cdot \alpha^{i}$ is also 
in $Q \setminus S$ at step $i$ (but not necessarily in a single state). 

By definition of $S$, and using an argument similar to the case
of always synchronizing, for all $i \geq 0$,
there exists a length $\ell_i \leq n$ such that 
a probability mass at least $\alpha_0 \cdot \alpha^{i} \cdot \frac{\alpha^{n}}{n}$
is in $Q \setminus T$ at step $i+\ell_i$.
We construct the indices $i_0, i_1, \dots$ inductively as follows. 
Let $i_0 = \ell_0$ 
and for all $j \geq 0$ let $i_{j+1} = i_{j} + 1 + \ell_{i_{j}+1}$.
Then $\M^{\straa}_{i_j}(T) < 1$ is not $1$-synchronized in $T$
(in fact $\M^{\straa}_{i_j}(T) \leq 1 - \alpha_0 \cdot \alpha^{i_{j-1}+n}$)
for all $j\geq 0$.
Note that $i_0 \leq n$ and $0 < i_{j+1} - i_{j} \leq n$ for all $j\geq 0$.

\item 
if $d_0 \not\in \winas{strongly}(T)$, then
$d_0$ is not almost-sure winning for the reachability objective $\Diamond S$,
and by Lemma~\ref{lem:almost-sure-reach-MDP} 
for all (randomized) strategies $\straa$ and all $i \geq 0$,
there exists a state $q_i \in Q \setminus S$ such that 
$\Pr^{\straa}_{d_0}(\Diamond^{i} \{q_i\}) \geq \alpha_0 \cdot \frac{\alpha^{n}}{n}$.
Since $q_i \not\in S$ is not in the winning region for sure safety $\Box T$,
there is a length $\ell_i \leq n$ as above, 
such that the probability mass in $Q \setminus T$ at length $i+\ell_i$
is at least $\alpha_0 \cdot \frac{\alpha^{2n}}{n^2}$.
For the indices $i_0 = \ell_0$ and $i_{j+1} = i_{j} + 1 + \ell_{i_{j}+1}$
for all $j \geq 0$, we get $\M^{\straa}_{i_j}(T) \leq 1 - \alpha_0 \cdot \frac{\alpha^{2n}}{n^2}$. 
\end{itemize}
\end{longversion}

\begin{theorem}\label{theo:epsilon-always-strongly}
Given an initial distribution $d_0$,
let $\alpha_0$ be the smallest positive probability in~$d_0$,
and let $\epsilon_a = \alpha_0 \cdot \frac{\alpha^{n}}{n}$ and $\epsilon_s = \alpha_0 \cdot \frac{\alpha^{2n}}{n^2}$.

\begin{itemize}
\item 
if $d_0 \not\in \winsure{always}(T)$ is not sure always synchronizing in~$T$, 
then for all strategies $\straa$, in the sequence $\M^{\straa}$ 
there exists a position $i \leq n$ such that $\M^{\straa}_{i}$ is not 
$(1-\epsilon_a)$-synchronized in $T$,

\item 
if $d_0 \not\in \winsure{strongly}(T)$ is not sure strongly synchronizing in~$T$, 
then for all strategies $\straa$, in the sequence $\M^{\straa}$ 
there exist infinitely many positions $i_0 < i_1 < i_2 < \dots$ 
where $i_0 \leq n$ and $i_{j+1} - i_{j} \leq n$ for all $j\geq 0$
such that $\M^{\straa}_{i_j}$ is not $1$-synchronized in $T$.

\item if $d_0 \not\in \winas{strongly}(T)$ is not almost-sure strongly synchronizing in~$T$, 
then for all strategies $\straa$, in the sequence $\M^{\straa}$ 
there exist infinitely many positions $i_0 < i_1 < i_2 < \dots$ 
where $i_0 \leq n$ and $i_{j+1} - i_{j} \leq n$ for all $j\geq 0$
such that $\M^{\straa}_{i_j}$ is not $(1-\epsilon_s)$-synchronized in $T$.
\end{itemize}
\end{theorem}

\section{Adversarial Synchronizing Objectives}
In an adversarial MDP the strategies are universally quantified,
which corresponds to satisfying an objective regardless
of the choice of strategies by an adversary.
Replacing $\exists \straa$ by $\forall \straa$ in the definition
of the three winning modes gives, after taking the negation
to get existentially quantified strategies, 
the following new winning modes.

Given a set $T \subseteq Q$,
we say that a sequence $d_0 d_1 \dots$ of probability distributions is:
\begin{itemize}
\item \emph{positively} \{always, eventually, weakly, strongly\} winning 
if $d_i(T) > 0$ for, respectively, all $i \geq 0$, some $i \geq 0$, infinitely many $i$'s,
all but finitely many $i$'s.

\item \emph{boundedly}  \{always, eventually, weakly, strongly\} winning 
if there exists $\epsilon>0$ 
such that $d_i(T) > \epsilon$ for, respectively, all $i \geq 0$, some $i \geq 0$, infinitely many $i$'s,
all but finitely many $i$'s.
\end{itemize}

For $\lambda \in \{always, event, weakly, strongly\}$, 
we denote by $\winpos{\lambda}(T)$ (resp., $\winbound{\lambda}(T)$)
the set of initial distributions $d_0$ from which 
there exists a strategy $\straa$ such that
the sequence $\M^{\straa}$ is positively (resp., boundedly) $\lambda$-synchronizing in $T$,
and we say that $\straa$ is positively (resp., boundedly) $\lambda$-synchronizing in $T$
from $d_0$. 


Table~\ref{tab:def-modes} summarizes the new definitions.
Note that replacing the existential quantification on strategies
in boundedly winning mode by a supremum gives the same question, 
since $\exists \straa: f(\straa) > 0$ is equivalent to $\sup_\straa f(\straa) > 0$. 
For the same reason, we have the identity $\winpos{event}(T) = \winbound{event}(T)$.
It is easy to show that the definitions imply the identity 
$\winbound{always}(T) = \winpos{always}(T) \cap \winbound{strongly}(T)$,
which we also obtain as a corollary of Lemma~\ref{lem:pos-bounded} below.

\begin{remark}\label{rem:expressiveness}
It immediately follows from the definitions that 
for all synchronizing modes $\lambda \in \{always, event, weakly, strongly\}$, 
and $\mu \in \{positive, bounded\}$:
\begin{itemize}
\item $\win{always}{\mu}(T) \subseteq \win{strongly}{\mu}(T) \subseteq \win{weakly}{\mu}(T) \subseteq \win{event}{\mu}(T)$,
\item $\winbound{\lambda}(T) \subseteq \winpos{\lambda}(T)$, 
\end{itemize}
and moreover,
\begin{itemize}
\item $\winpos{event}(T) = \winbound{event}(T)$, and
\item $\winbound{always}(T) = \winpos{always}(T) \cap \winbound{strongly}(T)$.
\end{itemize}
\end{remark}

It is easy to see that if there exists a strategy $\straa$ that 
is positively $\lambda$-synchronizing in $T$, then the
strategy $\straa_{\u}$ that plays at every round all actions uniformly 
at random is also positively $\lambda$-synchronizing in $T$,
because the condition $d_i(T) > 0$ is equivalent to 
$\Supp(d_i) \cap T \neq \emptyset$, and because we have $\straa \subseteq \straa_{\u}$,
which implies that 
$\Supp(\M_i^{\straa}) \subseteq \Supp(\M_i^{\straa_{\u}})$
for all $i \geq 0$.

Hence, in all four synchronization modes, 
the question is equivalent to the same question in Markov chains
(obtained from the given MDP by fixing the strategy $\straa_{\u}$)
which can be solved as follows.
Given a Markov chain, consider the underlying directed graph $\tuple{Q,E}$
where $(q,q') \in E$ if $\delta(q,a)(q') > 0$ (where $\Act = \{a\}$).
For positive eventually synchronizing, it suffices to find a state in $T$ 
that is reachable in that graph, and for positive weakly synchronizing,
it suffices to find a state in $T$ that is both reachable and can reach itself.
These questions are NL-complete.
For positive always and strongly synchronizing, the question is equivalent
to the model-checking problem for the formulas $G \exists T$ 
and $FG \exists T$ in the logic CTL+Sync, which are both 
coNP-complete~\cite[Lemma~2 \& Section~3]{CD16}.

For boundedly winning, we show that one strategy is good enough
in all four synchronization modes, like for positive winning. 
The strategy plays like $\straa_{\u}$ for the first $2^n$ rounds,
and then switches to a strategy $\straa_{\Epsilon}$ that, 
in the states $q \in \Epsilon$, plays uniformly at random 
all actions that stay in the end-component $\Epsilon(q)$ of $q$
(thus all actions in $\Act_{\Epsilon(q)}$),
and in the transient states $q \not\in \Epsilon$, plays 
all actions uniformly at random. We call this strategy
the \emph{freezing} strategy. Intuitively we use $\straa_{\u}$
to scatter the probability mass in all end-components of the MDP,
and then $\straa_{\Epsilon}$ to maintain a bounded probability
in each end-component.


\begin{lemma}\label{lem:pos-bounded}
Let $\M$ be an MDP with $n$ states and initial distribution $d_0$, and let $T$ be a target set. 
Consider the following conditions:
\begin{itemize}
\item[(1)] $\forall i \geq 0: \M^{\straa_{\u}}_i(T) > 0$  
\quad \quad (2) $\forall i \geq 2^n: \M^{\straa_{\u}}_i(\Epsilon \cap T) > 0$
\end{itemize}
Then, the following equivalences hold:

\begin{itemize}
\item[(a)] $d_0 \in \winpos{always}(T)$ if and only if Condition~$(1)$ holds;  
\item[(b)] $d_0 \in \winbound{strongly}(T)$ if and only if Condition~$(2)$ holds;  
\item[(c)] $d_0 \in \winbound{always}(T)$ if and only if Conditions~$(1)$ and~$(2)$ hold;  
\end{itemize}
\end{lemma}

\begin{proof}
Equivalence~$(a)$ follows from the definition of positive always synchronizing,
and from the fact that the uniform strategy $\straa_{\u}$ is sufficient for positive 
winning.

We show Equivalence~$(b)$ as follows. 
First, if Condition~$(2)$ does not hold, then 
$\M^{\straa_{\u}}_i(\Epsilon \cap T) = 0$ for some $i \geq 2^n$,
and thus also for infinitely many $i$'s (since the sequence $\Supp(\M^{\straa_{\u}}_i)$
is ultimately periodic, after at most $2^n$ steps). 
For arbitrary strategy $\straa$, we have $\Supp(\M^{\straa}_i) \subseteq \Supp(\M^{\straa_{\u}}_i)$
for all $i \geq 0$, therefore $\M^{\straa}_i(\Epsilon \cap T) = 0$ for 
infinitely many $i$'s.
By Lemma~\ref{lem:ec}, we have $\liminf_{i \to \infty} \M^{\straa}_{i}(\Epsilon) = 1$
which entails that $\limsup \M^{\straa}_{i}(\Epsilon \setminus T) = 1$
and $\limsup \M^{\straa}_{i}(Q \setminus T) = 1$, 
that is $\liminf \M^{\straa}_{i}(T) = 0$.
Since this holds for arbitrary strategy $\straa$, 
it follows that $d_0 \not\in \winbound{strongly}(T)$.

For the converse direction, assuming Condition~$(2)$ holds, we show that 
$d_0 \in \winbound{strongly}(T)$, witnessed by the freezing strategy~$\straa_f$
(which plays like $\straa_{\u}$ for the first $2^n$ rounds,
and then switches to the strategy $\straa_{\Epsilon}$).

We show the key property that 
$$\Supp(\M^{\straa_{\u}}_i) \cap \Epsilon = \Supp(\M^{\straa_f}_i) \cap \Epsilon 
\text{ for all } i \geq 2^n.$$

Fix an arbitrary $i \geq 2^n$ and let $p$ be a period of the sequence 
$\Supp(\M^{\straa_{\u}})$ such that $i-p \leq 2^n$ and
$\Supp(\M^{\straa_{\u}}_i) = \Supp(\M^{\straa_{\u}}_{i-p})$.
Consider the Markov chain $\M_{\Epsilon}$ obtained by 
fixing the strategy $\straa_{\Epsilon}$ in~$\M$. 
%
From the basic theory of Markov chains, each
end-component~$C$ in $\M$ is a recurrent class in $\M_{\Epsilon}$.
For each $i \geq 2^n$, either all or none of the states in a 
periodic class of~$C$ are in the support of $\M^{\straa_{\u}}_i$. 

To show the key property, first consider a state $q \in \Supp(\M^{\straa_{\u}}_i) \cap \Epsilon$
for $i \geq 2^n$,
and show that $q \in \Supp(\M^{\straa_f}_i) \cap \Epsilon$.
Let $S$ be the periodic class of $\Epsilon(q)$ containing $q$ (in $\M_{\Epsilon}$), and thus 
$$S \subseteq \Supp(\M^{\straa_{\u}}_i) \cap \Epsilon, \text{ and thus }
S \subseteq \Supp(\M^{\straa_{\u}}_{i-p}) \cap \Epsilon.$$
Since $\straa_{\u}$ and $\straa_f$ coincide on the first $2^n$ rounds, 
we have $$S \subseteq \Supp(\M^{\straa_{f}}_{i-p}) \cap \Epsilon.$$

Now consider the strategy $\straa_{\Epsilon}$ and an initial distribution
with support $S$, and denote by $S+j$ the support of the probability
distribution after playing $\straa_{\Epsilon}$ for $j$ steps. Then,
since $\straa_{\Epsilon} \subseteq \straa_{f} \subseteq \straa_{u}$, 
$$S+p \subseteq \Supp(\M^{\straa_{f}}_{i}) \cap \Epsilon, \text{ and }
S+p \subseteq \Supp(\M^{\straa_{u}}_{i}) \cap \Epsilon.$$

We can repeat the same argument with $S+p$ instead of $S$, 
and show by induction that $S+ j \cdot p \subseteq \Supp(\M^{\straa_{f}}_{i}) \cap \Epsilon$
for all $j \geq 1$. In particular, by taking $j$ the period of the 
end-component containing $q$, we get $S+ j \cdot p = S$ and thus
$S \subseteq \Supp(\M^{\straa_{f}}_{i}) \cap \Epsilon$, which establishes
one direction of the key property (the converse direction follows from 
$\straa_f \subseteq \straa_{\u}$).

From the theory of Markov chains, in every end-component state $q \in \Epsilon$,
the positive probability mass is bounded away from~$0$ in $\M^{\Epsilon}$, 
that is there exists a bound $\epsilon > 0$ such that 
for all $i \geq 2^n$, for all $q \in \Epsilon$, if $\M^{\straa_{f}}_i(q) \neq 0$,
then $\M^{\straa_{f}}_i(q) \geq \epsilon$. 
By the key property and Condition~$(2)$,
for all $i \geq 2^n$, there exists $q \in \Epsilon \cap T$ such that
$\M^{\straa_{f}}_i(q) \neq 0$, which implies that 
$\liminf_{i \to \infty} M^{\straa_{f}}_i(T) \geq \epsilon > 0$ 
and thus $d_0 \in \winbound{strongly}(T)$.

Finally, the proof for Equivalence~$(c)$ follows the same steps as above to
show that Conditions~$(1)$ and~$(2)$ imply $d_0 \in \winbound{always}(T)$,
where Condition~$(1)$ is used to bound $M^{\straa_{f}}_i(T)$ for the first
$2^n$ rounds, and thus to ensure that $M^{\straa_{f}}_i(T) \geq B > 0$
for all $i \geq 0$, hence $d_0 \in \winbound{always}(T)$.
The converse direction immediately follows from the first part
of Remark~\ref{rem:expressiveness} and Equivalences~$(a)$~and~$(b)$. \qed
\end{proof}

We extract explicit bounds from the proof of Lemma~\ref{lem:pos-bounded}.
All end-components are reached within a most $n$ steps (under $\straa_{\u}$),
and further all states in (a periodic class of) a recurrent class
are reached (synchronously) within a most $n^2$ steps~\cite[Theorem~4.2.11]{Gallager},
thus all states in the periodic class have probability mass at least 
$\epsilon_a = \alpha_0 \cdot \left(\frac{\alpha}{\abs{\Act}}\right)^{n+n^2}$
where $\alpha_0$ is the smallest positive probability in the initial
distribution~$d_0$ (note that $\alpha/\abs{\Act}$ is the smallest
probability in the Markov chain $\M_{\Epsilon}$). 
It follows that the freezing strategy ensures
probability at least $\epsilon_a$ in $T$ at every step (if $\M$ is boundedly
always synchronizing), and probability at least $\epsilon_a$ in $T$ at 
every step after $N = n + n^2$.

The conditions~$(1)$ and~$(2)$ in Lemma~\ref{lem:pos-bounded} can be decided 
in coNP as follows. For Condition~$(1)$ we guess an index $i \leq 2^n$ (in 
binary) and compute the $i$-th power of the Boolean transition matrix 
$M \in \{0,1\}^{n^2}$ where $M(q,q') = 1$ if there is a transition from 
state~$q$ to state~$q'$ in the Markov chain obtained from the given MDP $\M$
by fixing the strategy $\straa_{\u}$. The matrix $M^i$ 
can be computed in polynomial time by successive squaring of $M$.
Then it suffices to check whether $M^i(q_0,q) = 0$ for all $q_0 \in \Supp(d_0)$
and $q \in T$. For Condition~$(2)$, since the sequence $\Supp(\M^{\straa_{\u}}_i)$
is ultimately periodic, we guess two indices $i,p \leq 2^n$ ($p \geq 1$) and check 
that $\Supp(\M^{\straa_{\u}}_i) = \Supp(\M^{\straa_{\u}}_{i+p})$ and
$\Supp(\M^{\straa_{\u}}_i) \cap \Epsilon \cap T = \emptyset$, using 
the same approach by successive squaring. Note that the union $\Epsilon$ 
of all end-components can be computed in polynomial time~\cite{CY95,deAlfaro97}.

Conditions~$(1)$ and~$(2)$ are also coNP-hard, using the same reduction 
that established coNP-hardness of the positive always and positive bounded
synchronizing~\cite[Lemma~2 \& Section~3]{CD16}, in which positive and bounded
winning mode coincide.
It follows that the membership problem for bounded always and bounded strongly
synchronizing is coNP-complete.

We now show the solution for bounded weakly synchronizing.
It suffices to find a state in $T \cap \Epsilon$ that is reachable 
in the underlying graph of the Markov chain $\M_{\straa_{\u}}$,
which is a NL-complete problem (like for positive weakly synchronizing,
except we require a reachable state in $T \cap \Epsilon$, not just in $T$).
Indeed, if all reachable end-components are contained in $Q \setminus T$,
then by Lemma~\ref{lem:ec} we have 
$\liminf_{i \to \infty} \M^{\straa}_{i}(Q \setminus T) = 1$,
that is $\limsup_{i \to \infty} \M^{\straa}_{i}(T) = 0$.
For the converse direction, if a state $\hat{q} \in T \cap \Epsilon$
is reachable, then by a similar argument as above
based on the theory of Markov chains, as the probability
mass in the states of the periodic classes (that contain some
probability mass) is bounded away from~$0$ in $\M^{\straa_{\Epsilon}}$,
it follows that within every $p$ steps, where $p$ is the period of the recurrent
class $\Epsilon(\hat{q})$ the probability mass in $\hat{q}$ is at least 
$\alpha_0 \cdot (\alpha / \abs{\Act})^{n+n^2}$.
Therefore, $\M$ is boundedly weakly synchronizing in $T$.
For the sake of completeness, note that for eventually synchronizing MDPs, 
the probability mass $\epsilon_e = \alpha_0 \cdot (\alpha/ \abs{\Act})^{n}$
in $T$ can be ensured within $n$ steps (using~$\straa_{\u}$).

\begin{theorem}\label{theo:pos-bound-complexity}
The complexity of the membership problem for positive and bounded 
synchronizing objectives is summarized in Table~\ref{tab:complexity}. 
\end{theorem}
 
In Table~\ref{tab:complexity}, the merged cells for eventually 
synchronizing reflect the fact that the winning regions coincide (see
Remark~\ref{rem:expressiveness}). The winning regions for the other 
synchronizing modes do not coincide, already in Markov chains (see \figurename~\ref{fig:mc}).

\paragraph{{\bf Acknowledgment.}}
The authors are grateful to Jean-Fran\c{c}ois Raskin for logistical support, 
and to Mahsa Shirmohammadi for interesting discussions about adversarial objectives.

\bibliographystyle{abbrv}
\bibliography{biblio} 


\end{document}